\documentclass[runningheads]{llncs}
\usepackage[T1]{fontenc}
\usepackage{amsmath, amssymb}
\usepackage{graphicx}
\usepackage{silence}
\usepackage{tikz}
\usetikzlibrary{matrix}
\usetikzlibrary{positioning}

\usepackage[colorlinks=true, allcolors=blue]{hyperref}
\usepackage[hyphenbreaks]{breakurl}

\usepackage{algorithm}
\usepackage{pseudo}
\pseudoset{
	ctfont=\color{black!75},
	ct-left=\unskip\qquad\texttt{$\triangleright$ },
	ct-right=\texttt{}
}

\renewcommand\subsubsection[1]{\textbf{#1.}~}

\def\PUT/{\cn{put}}
\def\FOUND/{\cn{found}}
\def\FULL/{\cn{full}}
\def\EMPTY/{\cn{empty}}
\def\TRUE/{\cn{true}}
\def\FALSE/{\cn{false}}
\def\rk{\id{rk}}
\def\load{\id{load}}

\def\cas{\pr{Cas}}
\def\Any{\pr{Any}}
\def\First{\pr{First}}
\def\Count{\pr{Count}}
\def\Shuffle{\pr{Shuffle}}
\newcommand{\Fop}{\pr{Fop}}

\newcommand\concept[1]{\textit{#1}}

\newcommand{\concat}[2]{\ensuremath{#1 . #2}}

\DeclareMathOperator{\sg}{sg}

\title{Compact Parallel Hash Tables on the GPU}
\author{Steef Hegeman \inst{1}\orcidID{0009-0009-0368-8502} \and
Daan Wöltgens\inst{2,3} \and
Anton Wijs\inst{2}\orcidID{0000-0002-2071-9624} \and
Alfons Laarman\inst{1}\orcidID{0000-0002-2433-4174}}
\authorrunning{S. Hegeman et al.}
\institute{LIACS, Leiden, The Netherlands\\
\email{\{s.hegeman,a.w.laarman\}@liacs.leidenuniv.nl} \and
Eindhoven University of Technology, Eindhoven, The Netherlands\\
\email{a.j.wijs@tue.nl} \and
TNO, The Netherlands\\
\email{daan.woltgens@tno.nl}}

\begin{document}
\maketitle

\begin{abstract}
	On the GPU,
	hash table operation speed
	is determined in large part by cache line efficiency,
	and state-of-the-art hashing schemes thus divide tables into cache line-sized buckets.
	This raises the question whether performance can be further improved
	by increasing the number of entries that fit in such buckets.
	Known compact hashing techniques
	have not yet been adapted to the massively parallel setting,
	nor have they been evaluated on the GPU.
	We consider a compact version of bucketed cuckoo hashing,
	and a version of compact iceberg hashing
	suitable for the GPU.
	We discuss the tables from a theoretical perspective,
	and provide an open source implementation
	of both schemes in CUDA for comparative benchmarking.
	In terms of performance,
	the state-of-the-art cuckoo hashing benefits from compactness on lookups and insertions
	(most experiments show at least 10--20\% increase in throughput),
	and the iceberg table benefits significantly,
	to the point of being comparable to compact cuckoo hashing---%
	while supporting performant dynamic operation.

\end{abstract}

\section{Introduction}
General purpose graphics processing units (GPUs)
have been used to significantly speed up computations in many different domains.
With thousands of processing cores,
GPUs offer access to massive parallelism for everyone.
However, the main GPU memory tends to be scarcer than the main memory available to CPUs.
Moreover, memory access times are often the main performance bottleneck of GPU programs.
Therefore, data structures that sparingly use GPU memory
can positively affect both the amount of data that can be stored and the performance of a GPU program.

Memory-efficiency can be achieved by \concept{quotienting},
a technique for reducing the storage required per key by using its storage location as information.
This was first used in practice in Cleary's compact hash table~\cite{cleary_compact_1984}.
This reduces the memory usage per table slot logarithmically in the number of slots.
The parallelization of Cleary's hashing scheme in \cite{vegt_parallel_2012} involves coarse locking.
On the GPU, the (coarse) locking strategies of traditional CPU multicore algorithms do not perform well.
The fastest GPU tables use atomic operations to directly insert keys into the table.
We refer to this as \concept{lockless}.
In addition, it is fair to say that optimizing GPU hash table performance
is about reducing the number of cache lines involved per operation \cite{awad_analyzing_2023}.

Among the available GPU hash tables,
adaptations of cuckoo hashing \cite{pagh_cuckoo_2004,alcantra_building_2012,awad_analyzing_2023}
and iceberg hashing \cite{bender_iceberg_2023,awad_analyzing_2023} have proven promising.
In particular, the bucketed cuckoo table of \cite{awad_analyzing_2023}, the slab hash table of \cite{ashkiani_dynamic_2018} and DyCuckoo \cite{li_dycuckoo_2021} are state-of-the-art.
Indeed, \cite{arbitman_backyard_2010,bender_iceberg_2023} confirm that at least in theory compact hashing with cuckoo or iceberg hashing is a good idea.
A downside of these tables is however that they are more strict in terms of where a key can be stored.
With cuckoo hashing, this can lead to situations where there are still empty slots in the table,
but specific keys cannot be added to it.
So here, one should not just consider the memory savings per slot,
but also the expected \concept{fill factor} of the table.

Moreover, not all hash tables offer the same guarantees.
The bucketed cuckoo table of \cite{awad_analyzing_2023} for example is \concept{static},
meaning that it only supports a mode of operation where the table is built once
out of a batch of unique keys,
and can then only be queried.
As inserts temporarily move existing keys from the table to local memory,
lookup operations performed during insertion may return false negatives,
and concurrent insertion operations may cause keys to be stored in multiple slots.
For the same reason,
DyCuckoo~\cite{li_dycuckoo_2021}, while supporting resizing,
does not support concurrent insert operations.
We say that a table supports \concept{dynamic} operation
if it supports concurrent combinations of lookups and writes.
In many applications,
dynamic operation
is essential.
For instance, in model checking~\cite{principlesofmodelchecking,wijs_model_checking_2023},
fixpoint algorithms continuously check whether keys have been seen before and insert them if not.

We propose a compact, lockless GPU hash table that supports concurrent inserts.
Our table is based on iceberg hashing, with a provably correct find-or-put operation.
We provide an open source implementation in CUDA.
We also combine the lockless GPU hash table based on
cuckoo hashing from~\cite{awad_analyzing_2023} with compact hashing for comparison.
Synthetic benchmarks show that
both tables benefit from compactness
(10--20\% speedup on lookups and insertions),
while also halving or nearly quartering memory usage in many situations.
Furthermore, the compact iceberg table performance is very close
to the compact cuckoo one in static situations%
---a significant result,
as the cuckoo table of \cite{awad_analyzing_2023}
is, to the best of our knowledge,
the fastest static GPU hash table to date,
and the compact cuckoo table is comparable to it.
This demonstrates the competitiveness of the compact iceberg table with the state-of-the-art.

In conclusion, we establish that when it comes to hash tables on the GPU,
you can have your cake and eat it too:
compact hashing can lead to both reduced memory usage as well as improved performance.

\section{Background}
\subsection{Hash tables}
A hash table $T$ is a data structure implementing a set or map from keys $K$ to values $V$.
Some hash tables are \concept{stable}, meaning that the index at which a key is stored does not change.
In this paper, we focus on hash tables representing sets,
but the presented operations can be extended to a map implementation,
by storing key-value or key-pointer pairs instead of keys, or, in the case of a stable table,
by storing values in a separate array at the same index as the key.

Hash tables are typically optimized for \pr{Find}, \pr{Put}, and \pr{Delete} operations.
We generally do not consider deletions (or treat them as if exceptional),
as we are interested in the use of hash tables in search algorithms.

In this paper,
our tables generally consist of an array of \emph{buckets}.
Each bucket contains one or more \concept{slots}.
Each slot is either unoccupied (has value $\EMPTY/$) or stores a key,
possibly together with additional bookkeeping information.
$H$ hash functions $h_0,\dots,h_{H-1}: K \rightarrow T$  map keys to buckets in $T$.
Given a hash function $h_i$,
a key $k$ can be stored in some slot with index $j$ of the corresponding bucket $T[h_i(k)]$, i.e., 
$T[h_i(k)][j]$.

A table's \concept{fill factor} is defined as the fraction of slots that are occupied.
If key $k$ is to be inserted into the table,
but buckets indexed by $h_0(k),\dots,h_{H-1}(k)$ are full
(all their slots are occupied),
then there is typically no way to add $k$ to the table and the  table is considered full
(an exception is the cuckoo hashing scheme below).
A bigger table can then be allocated to store the set (\concept{rehashing}) \cite{pagh_cuckoo_2004}.
The larger $H$,
the greater the expected fill factor a table achieves before rehashing is needed,
but the longer lookups may take, as potentially all buckets $h_i(k), i < H$ need to be checked.

In which of the buckets $h_i(k)$ a key $k$ is stored,
and how the table is manipulated,
is determined by the hashing \concept{scheme}.
The memory efficiency of a scheme depends on the expected fill factor that can be achieved,
as well as the memory used per table slot.
We discuss two schemes below.

\subsection{Cuckoo hashing}
\subsubsection{Insertion}
Cuckoo hashing \cite{pagh_cuckoo_2004,panigrahy_efficient_2005,erlingsson_cool_2006}
achieves high fill factors in practice
while limiting the slots in which a key may be placed,
by moving keys between buckets during insertion.
A yet to be inserted key may take the place of a key already in the table
(\emph{evicting} the original key),
forcing it to be moved elsewhere.

Each slot either contains a dedicated value $\EMPTY/$,
signaling it is unoccupied,
or (in the case of cuckoo hashing) stores a pair $(k,j)$ of key and hash function index $j < H$.
When a slot containing $(k,j)$ is evicted,
the key $k$ is in turn directed to bucket $h_{j+1 \bmod H}(k)$,
evicting a slot there,
until finally a key is directed to a bucket with an empty slot,
in which case it takes the empty slot and the insertion is finished.
If the chain of evictions exceeds a threshold length $C$,
the insertion is aborted and the table is considered full.
An optimization is to, instead of storing a hash function index with the key,
recover it from the key $k$ found in bucket $b$
as the first $j$ with $h_j(k) = b$.

\subsubsection{Lookup}
A key $k$ is only found in a bucket $h_{i+1}(k)$ if it was kicked out of bucket $h_{i}(k)$,
so if deletions are forbidden,
the $\pr{Find}(k)$ procedure only has to check indices $h_0(k),h_1(k),\dots$ in order up to
the first bucket containing $k$ or an unoccupied slot.
(For rare deletion operations, additional `tombstone' values can preserve this invariant,
as discussed in, for example, \cite{gunji_studies_1980}.)

\subsection{Iceberg hashing}
Iceberg hashing \cite{bender_iceberg_2023} divides the table into multiple \emph{levels}.
Each key is assigned to one large \emph{primary bucket} in the first level by $h_0(k)$,
and two smaller \emph{secondary buckets} in the second level indicated by $h_1(k),h_2(k)$.
There is an additional third level outside of the table structure, made up of linked lists.
Iceberg hashing is stable:
once a key is inserted into a slot,
it is never moved to another \cite{bender_iceberg_2023}.

\subsubsection{Insertion}
On insertions,
a key $k$ is placed in the primary bucket if it has unoccupied slots,
otherwise it is placed in its secondary bucket with the most unoccupied slots.
This choice aspect has significant impact on load balancing \cite{azar_balanced_1994}.
If both the primary and secondary buckets of $k$ are full,
then it gets sent to the third level,
where $k$ is inserted into the linked list indicated by $h_0(k)$.
As level 3 can grow arbitrarily large,
insertions can never fail.
But to maintain performance,
the table should be rehashed when the third level grows large.

\subsubsection{Lookup}
Again assuming deletions are prohibited, lookups can be performed as follows.
First inspect the primary bucket.
If it contains $k$, we are done,
and if it does not contain $k$ but some of its slots are unoccupied, we are done as well.
Otherwise, both secondary buckets need to be inspected.
If $k$ is found in neither and one of the buckets has an unoccupied slot,
then $k$ is not in the table.
If all three buckets are full and do not contain $k$,
the linked list indicated by $h_0(k)$ needs to be searched for $k$.

\subsubsection{Level sizes}
Not all buckets in the table have the same number of slots:
the primary buckets are generally chosen to have more slots than the secondary ones \cite{pandey_iceberght_2023}.
Additionally, there are typically fewer secondary buckets than primary buckets \cite{pandey_iceberght_2023}.
So it might be instructive to think of an iceberg table as two hash tables and an array of linked lists.

The multi-level approach allows iceberg tables to achieve a high fill factor.%
\footnote{
	In \cite{bender_iceberg_2023}, definition of fill factor (space efficiency) is adapted to include level 3.
}
On the other hand,
the more levels an operation needs to traverse,
the longer it takes.
If an operation only needs to work on the first two levels,
it runs in constant time,
but if level 3 also needs to be inspected, it does not.
With proper tuning of the parameters,
the expected spillage of level 1 to level 2 is small,
and the spillage of level 2 to level 3 even smaller \cite{bender_iceberg_2023,pandey_iceberght_2023}.
In \cite{pandey_iceberght_2023},
they use large primary buckets (64 slots),
and a 1 to 8 ratio in the size of the primary and secondary level.
In practice, iceberg hashing achieves a high fill factor
while maintaining near constant time performance (constant-time with high probability).

\subsection{Compact hashing}
In \cite{cleary_compact_1984},
Cleary
introduced a technique for compact hashing
(now known as \concept{quotienting} \cite{bender_iceberg_2023}),
which reduces the storage space required per slot
logarithmically in the number of buckets as follows.

Let $T$ be a hash table with $2^N$ buckets,
let $K = 2^M$ be the binary strings of length $M$,
and let $\pi: K \to K$ be a permutation of keys.
Assign key $k$ to the bucket $a(k)$,
where $a(k)$ is the number denoted by the first $N$ bits of $\pi(k)$.
This is also called the \emph{address} of $k$.
The remaining bits of $\pi(k)$,
called the \emph{remainder} $r(k)$, are then stored in a slot in this bucket.
The key occupying a slot can thus be recovered by combining
the address of the bucket that the slot is in and the remainder stored in the slot,
followed by computing the inverse under $\pi$.
In summary, $k = \pi^{-1}(\concat{a(k)}{r(k)})$, where $\concat{}{}$ stands for concatenation of strings.

The remainders are only $M-N$ bits in length
so per-slot and thus per-table memory savings are logarithmic in the total number of buckets.

Schemes using multiple hash functions $h_i$ to assign keys to multiple buckets
can be made compact by using multiple permutations $\pi_i$,
giving rise to multiple address functions $a_i$ and remainder functions $r_i$.

\section{A parallel compact iceberg hash table}
We provide a lockless parallel compact iceberg algorithm.
We focus on a parallel find-or-put operation which returns $\FOUND/$
if the given key is in the table, and otherwise inserts the key (or returns $\FULL/$).
With this operation, we can support ubiquitous fixpoint computations as the ones used in model checking~\cite{principlesofmodelchecking,wijs_model_checking_2023} and many other applications.
Both cuckoo and iceberg hashing have proven promising for GPU applications, as discussed in the introduction.
Because of the difficulty of realizing concurrent writes in cuckoo hashing, we opt for a compact iceberg table with a concurrent find-or-put operation.

\subsubsection{Find-or-put}
Algorithm~\ref{alg:iceberg-fop-abstract} gives our lockless parallel find-or-put procedure
for compact iceberg hashing, called \Fop.
We use a table $T_0$ for the primary buckets,
and $T_1$ for the secondary buckets,
each compacted separately,
so that $a_0(k)$ indicates a $k$'s primary bucket in $T_0$,
and $a_1(k),a_2(k)$ its secondary buckets in $T_1$.
The constants $B_0$ and $B_1$ are table parameters indicating
the number of slots per primary bucket and per secondary bucket respectively.

Lines in the algorithm are not considered to be atomic,
except for the com\-pare-and-swap operation $\pr{Cas}(a,b,c)$,
which checks if the value of $a$ is $b$,
and if so, replaces it with $c$,
all in one atomic operation.
It returns \TRUE/ if it did replace the value of $a$, and \FALSE/ otherwise.
(In particular,
if $c \neq \EMPTY/$ and multiple threads simultaneously execute $\cas(a,\EMPTY/,c)$ on the same empty slot $a$,
only one succeeds.)
The reads in lines~\ref{line:snapshot_b0},\ref{line:snapshot_b1},\ref{line:snapshot_b2}
are consequently atomic per-slot.

\begin{algorithm}
	\caption{Iceberg: lockless parallel find-or-put}
	\label{alg:iceberg-fop-abstract}
\begin{pseudo}[kw]*
	\hd{Fop}(k)\\
	while $\cn{true}$ \ct{Work on level 1} \\+
	$b \gets T_0[a_0(k)]$ \label{line:snapshot_b0} \ct{Create a local snapshot} \\
	if $(\exists i < B_0 \colon b[i] = r_0(k))$ return \cn{found} \ct{Found $k$, we are done} \label{line:found_0}\\
	if $(\forall i < B_0: b[i] \neq \EMPTY/)$ break \ct{Level 1 is full, go to level 2} \label{line:to-level-2}\\
	else \tn{$b$ has an empty slot---let $b[i]$ be the first} \label{line:first-1}\\+
		if $\cas(T_0[a_0(k)][i], \EMPTY/, r_0(k))$ return \PUT/
		\ct{Insertion attempt} \label{line:insertion-attempt-1} \\-
	\\-
	while $\cn{true}$ \ct{Work on level 2}\\+
	$b_1 \gets T_1[a_1(k)]$ \label{line:snapshot_b1}\\
	if $(\exists i < B_1: b_1[i] = (r_1(k), 0))$ return \cn{found} \label{line:found_1}\\
	$b_2 \gets T_1[a_2(k)]$ \label{line:snapshot_b2}\\
	if $(\exists i < B_1: b_2[i] = (r_2(k), 1))$ return \cn{found} \label{line:found_2}\\
	*&\ct{$k$ is not in the table, try to insert it into the least full secondary bucket}\\
	$i \gets$ $1$ if \tn{$b_1$ is strictly less full than $b_2$} else $2$ \label{line:least}\\
	if $(\forall u < B_1 : b_i[u] \neq \EMPTY/)$ return \cn{full} \label{line:full}\\
	else \tn{$b_i$ has an empty slot---let $b_i[j]$ be the first} \label{line:first-2}\\+
		if $\cas(T_1[a_i(k)][j], \EMPTY/, (r_i(k), i))$ return \PUT/
		\label{line:insertion-attempt-2}
		\ct{Insertion attempt}
\end{pseudo}
\end{algorithm}

\subsubsection{Correctness}
A proof of correctness is given in appendix~\ref{app:proof}.
The key idea is that when the algorithm attempts to insert into a slot,
all earlier slots are nonempty and have been read
(see e.g. lines~\ref{line:first-1},\ref{line:first-2}).
This allows the table to handle concurrent inserts with duplicate inputs.

\subsubsection{Limitations}
We limit ourselves to an iceberg table with level 1 and 2 tables%
---the difficulty in parallelizing iceberg hashing lies in the choice aspect of level 2.
It should be noted that \cite{awad_analyzing_2023}
does not truly separate the first and second level for its iceberg implementation,
defying the basis of the theoretical analysis in \cite{bender_iceberg_2023}.
We believe that the dynamic slab hash table of \cite{ashkiani_dynamic_2018}
could be adapted as a third level for the scheme below,
forming a full 3-level iceberg table.
We have found that already with the first two levels,
good fill factors can be achieved in practice.
We currently omit resizing and support for deletion operations,
as a vast number of operations can already be supported with the given find-or-put operation.
As discussed in for example \cite{gunji_studies_1980},
the table can be extended to support a delete operation using so-called  `tombstone' values,
which preserve the invariant that non-empty slots occur consecutively.
Using a stop-the-world approach, we believe resizing can also be implemented.
In most applications on the GPU, it suffices however to simply claim all memory for the task at hand.

\section{Implementation}
\subsection{Architectural considerations}
\subsubsection{Architecture}
We summarize the architecture of NVIDIA GPUs as described in the CUDA programming guide \cite{nvidia_cuda_guide}.
A GPU contains several multiprocessors,
each capable of executing multiple \emph{threads} (processes) in parallel.
Threads are divided in groups of 32 called \emph{warps}.
Warps are then assigned to multiprocessors.

As a consequence of this warp-oriented architecture,
if threads in a warp take different branches,
the execution of these branches may be serialized.
The highest performance is achieved if
all threads in a warp execute the same lines of code.
An upside of the tight coupling between threads in a warp
is that they can efficiently communicate,
and that their memory accesses can be \emph{coalesced}:
if threads in a warp access elements in the same cache line in parallel,
the line is retrieved from memory only once,
instead of once per thread.

In summary, to improve performance,
threads in a warp should have as little branch divergence as possible,
and should aim to access memory in the same cache line as often as possible.

\subsubsection{Cooperative work sharing}
For bucketed GPU hash tables,
this is typically achieved by \emph{warp-cooperative work sharing} \cite{ashkiani_dynamic_2018,awad_analyzing_2023}.
Each thread receives an input key,
but warps then cooperate,
working together on one of their threads' keys at a time.
When a bucket is inspected,
each thread in the warp reads one of the bucket slots,
together assessing the whole bucket in one
(or few, depending on the bucket size) coalesced reads.
CUDA allows for warps to be subdivided in smaller \concept{cooperative groups},
which can be used when buckets have fewer than 32 slots.
In summary,
cooperative work sharing allows all threads to do useful work
while decreasing the number of memory operations per thread.

\subsubsection{Group synchronization}
In CUDA, each thread has a global \concept{rank},
an index in the total number of threads.
In a cooperative group, each thread also has a local \concept{group rank},
from 0 to the group size (exclusive).
Cooperative groups have several synchronization primitives,
such as \texttt{shfl}, \texttt{any}, and \texttt{ballot},
which can be used to implement the following abstract procedures.

Let $G$ be a group.
$\Shuffle_{G,s}(v)$ evaluates $v$ in the thread with group rank $s$ and returns the result.
$\Any_G(P)$ is \cn{true} if and only if the predicate $P$
evaluates to \cn{true} in any thread in the group,
$\First_G(P)$ gives the group rank of the first thread in $G$ in which $P$ holds,
or $\bot$ if $P$ is not true in any,
and $\Count_G(P)$ gives the number of threads in $G$ in which $P$ holds.

\subsubsection{Global synchronization}
When one thread writes to memory
(in our case, a bucket slot),
there is no guarantee that this change is reflected in reads by other threads
until explicit synchronization,
unless this memory was written using atomic instructions.
(Even then, volatile loads are required.)
Of interest to us are \texttt{atomicCAS} and \texttt{atomicExch},
atomic $\cas$ and $\pr{Swap}$ operations, respectively.

\subsection{Iceberg find-or-put}
Algorithm~\ref{alg:iceberg-fop-concrete} describes the cooperative find-or-put procedure.
For simplicity, we assume that the
primary bucket size $B_0$ must divide 32 (the size of a warp),
and that the secondary buckets contain half that many rows.%
\footnote{
	The actual implementation also supports smaller secondary buckets.
	Larger primary buckets could be implemented.
}
Algorithm~\ref{alg:work-sharing}
implements the cooperative work sharing
for inputs of a multiple of $B_0$ keys.
The actual implementation supports input batches of any size.

\subsubsection{Implementation notes}
In the actual CUDA implementation,
the reads use volatile loads,
and there are some minor optimizations
(filled slots are not read again,
and it exploits that \texttt{atomicCAS} returns the read value of the target slot
to avoid rereading the slot after a failed insertion attempt).

\begin{algorithm}
	\caption{Iceberg find-or-put for cooperative group $G$ with $|G| = B_0 = 2B_1$}
	\label{alg:iceberg-fop-concrete}
\begin{pseudo}[kw]*
	\hd{CoopFop}(k, G)\\
	$\rk \gets \tn{my rank in $G$}$\\
	loop forever \\+
	$s \gets T_0[a_0(k)][\rk]$ \ct{Snapshot ``my'' bucket slot (coalesced read)}\\
	if $\Any_G(s = r_0(k))$ return \FOUND/\\
	$\id{load} \gets \Count(s \neq \EMPTY/)$ \ct{Compute the number of filled slots}\\
	if $\id{load} = B_0$ break \ct{Level 1 is full, go to level 2}\\
	else \ct{One of the threads tries to insert $k$ into the first empty slot}\\+
		if $\rk = \id{load}$\\+
			$s \gets \cas(T_0[a_0(k)][\rk], \EMPTY/, r_0(k))$\\-
		if $\Shuffle_{G,\id{load}}(s)$ return \PUT/ \ct{If it succeeds, we are done} \\-
	\\-
	\tn{subdivide $G$ into $G_1,G_2$ of even and odd threads respectively}\\
	$j \gets (\rk \bmod B_1) + 1$ \ct{Determine ``my'' subgroup}\\
	$\rk' \gets \tn{my rank in $G_j$}$\\
	loop forever \\+
	$s \gets T_1[a_j(k)][\rk']$ \ct{Subgroup $j$ inspects the $j$th bucket (coalesced reads)}\\
	if $\Any_G(r = (r_j(k), j-1))$ return \FOUND/ \\
	$\id{load} \gets \Count_{G_j}(s \neq \EMPTY/)$ \ct{Each subgroup calculates its bucket load}\\
	$\id{load}_1 \gets \Shuffle_{G,0}(\id{load})$\\
	$\id{load}_2 \gets \Shuffle_{G,1}(\id{load})$\\
	*&\ct{One thread tries to insert $k$ into the least full bucket}\\
	$i \gets$ $1$ if $\id{load}_1 < \id{load}_2$ else $2$\\
	if $\load_i = B_1$ return \FULL/\\
	else \\+
		if $\rk = i - 1$\\+
			$s \gets \cas(T_1[a_i(k)][\load_i], \EMPTY/, (r_i(k), i-1))$\\-
		if $\Shuffle_{G,i-1}(s)$ return \PUT/
\end{pseudo}
\end{algorithm}

\begin{algorithm}
	\caption{Find-or-put batch $\id{keys}$ of size $n = mB_0$}
	\label{alg:work-sharing}
\begin{pseudo}[kw]*
	\hd{Fop}(\id{keys}, n)\\
	$\rk \gets \tn{my global rank}$\\
	$k \gets \id{keys}[i]$ \label{line:work-share-coalesce}\\
	\tn{subdivide threads into groups of size $B_0$, let $G$ be my group}\\
	$b \gets \cn{true}$ \ct{I have work to do}\\
	$r \gets \FULL/$\\
	for \tn{$0 \leq i < B_0$}\\+
		$k' \gets \Shuffle_{G,i}(k)$\\
		$r' \gets \pr{CoopFop}(k', G)$\\
		if $\rk = i$\\+
			$r \gets r'$\\
			$b \gets \cn{false}$ \ct{I am done}\\--
	return $r$
\end{pseudo}
\end{algorithm}

\subsection{Permutations}
We use simple permutation functions:
one-round Feistel functions
based on the hash family used in \cite{awad_analyzing_2023}
for comparison purposes.
Users of the CUDA implementation can easily supply their own permutations.

\subsection{Parallel compact cuckoo implementation}\label{sec:cuckoo-impl}
We give a compact cuckoo implementation, that is close to
the bucketed cuckoo implementation of \cite{awad_analyzing_2023}.
The main difference being the use of permutations
instead of hash functions and the use of remainders,
cf. Algorithm~\ref{alg:compact-cuckoo-put}.

A downside of cuckoo hashing is that it is not stable:
keys move during insertions.
This makes it difficult to define a performant find-or-put procedure:
it could be that one process is checking whether $k$ is in the table
exactly when another process has decided to insert $k'$,
and has in this process temporarily evicted $k$ out of the table.
We see only one way to avoid such situations,
and that is to somehow synchronize all processes,
so that no process works on the insertion-phase of a find-or-put while others are in the lookup-phase.
However, on the GPU, synchronization of all parallel processes is very expensive.

Even if all processes are forced to finish their lookup-phase
before one starts their insertion phase,
there is yet another problem:
multiple processes executing find-or-put for the same key $k$.
After the lookup phase,
they could all conclude that $k$ has to be inserted into the table.
During the insertion phase,
after one process inserts $k$ into its first bucket,
an unrelated process might evict $k$ from this bucket in order to insert a key $k'$,
followed by one of the other processes inserting a new copy of $k$ in the first bucket.
This is another way in which multiple copies of a key could be inserted into the table.
While there could still be solutions to this
(locking, or using a prohibitive amount of metadata),
we do not see a performant parallel find-or-put algorithm for cuckoo hashing.
But as a static table, we can still make it compact.

\subsubsection{Insertion}
See Algorithm~\ref{alg:compact-cuckoo-put} for a parallel insertion algorithm for compact cuckoo.
Note the inversion in line~\ref{line:cuckoo-inv}
to recover the evicted key.
The $\cas$ operation is atomic as in Algorithm~\ref{alg:iceberg-fop-abstract}.
In addition the $\pr{Swap}(a,b)$ operation atomically swaps the values of $a$ and $b$.
Again, the snapshot read in line~\ref{line:cuckoo-snapshot} is atomic per-slot.

The CUDA implementation of Algorithm~\ref{alg:compact-cuckoo-put} uses the same warp-based work-sharing
approach as the iceberg table.

\begin{algorithm}
	\caption{Compact cuckoo: parallel put}
	\label{alg:compact-cuckoo-put}
\begin{pseudo}[kw]*
	\hd{Put}(k)\\
	$j \gets 0$\\
	for $c \in \{1,\dots, C\}$\\+
		$a \gets a_j(k)$\\
		$r \gets r_j(k)$\\
		$b \gets \tn{$T[a]$}$ \ct{Create a local snapshot} \label{line:cuckoo-snapshot}\\
		if $\exists i < B \colon b[i] = \EMPTY/$ \ct{If there is an empty slot}\\+
		if $\cas(T[a][i], \EMPTY/, (r,j))$ return \PUT/ \ct{Try to insert into it} \\-
		else\\+
			\tn{choose $i < B$} \ct{Pseudorandomly, or based on $c,k$} \\
			\tn{$\pr{Swap}$($(r,j)$, $T[a][i]$)} \ct{Atomically evict and claim a slot}\\
			$k \gets \pi_j^{-1}(\concat a r)$ \label{line:cuckoo-inv}
			\ct{We now continue with the evicted key}\\
			$j \gets j+1 \bmod H$ \\--
	return \FULL/
\end{pseudo}
\end{algorithm}

\subsubsection{Lookup}
Look for $r_i(k)$ in buckets $a_i(k)$.
As with non-compact cuckoo hashing,
the insertion guarantees that
if bucket $r_i(k)$ is not full,
then buckets $r_j(k)$ with $j > i$ are completely empty (buckets are filled in order).
Hence the search inspects the buckets for $k$ in order,
and stops early if a bucket $a_i(k)$ is inspected that does not contain $r_i(k)$ and is not full
(then $k$ is not in the table).

\subsection{CUDA code}
The full code,
used in the benchmarks below,
has been accepted to appear as a conference artifact \cite{artifact}.
The latest version of our CUDA library is open source and available at
\url{https://github.com/system-verification-lab/compact-parallel-hash-tables}.
The cuckoo part is based partially on the second author's master's thesis \cite{woltgens_cleary_2022}.
For simplicity,
we support only tables of which the number of slots (per level) is a power of 2,
and we assume input keys are 64 bits wide.
The code is set up so that
these restrictions can be eliminated if so desired.
Users can easily supply their own permutation functions.

The iceberg table uses \texttt{atomicCAS} (for $\cas$),
and the cuckoo table also uses \texttt{atomicExch} (for $\pr{Swap}$).
Both operations support word widths of 32, 64, and 128 bits,
and \texttt{atomicCAS} additionally supports 16-bit words.
Our cuckoo table thus supports slots of 32, 64 and 128 bits,
and the iceberg table additionally supports slots of 16 bits.

Smaller slot sizes could be supported by using atomic instructions that are ``too coarse'',
likely at a performance cost.
It is however important to be mindful of the global memory layout:
the total memory usage of a bucket should divide the cache line size (128 bytes)
lest some buckets end up being spread out over multiple cache lines,
degrading performance.

\section{Experimental evaluation}
\subsection{Synthetic benchmarks}
We want to measure the performance impact of compact hashing.
To this end, we set up tables (in various bucket combinations) with the same total number of slots.
Our original keys are wider than 32 bits,
but we can store them in slots of 32 (cuckoo) or 16, 32 (iceberg) bits.
We also benchmark the same tables with 64 bit slots
that would fit the original keys.
We can see the 64 bit tables as a baseline: they essentially behave as non-compact cuckoo and iceberg tables.%
\footnote{
	Microbenchmarks have shown that the runtime of computing the permutations themselves is negligible,
	so this is a fair assumption.
}
Thus, our benchmarks show the impact of compact versus non-compact hashing in terms of runtime performance.
Based on \cite{awad_analyzing_2023},
we expect the tables with 32 slots per (primary) bucket in particular to benefit significantly,
as the 64 bit versions use two cache lines per (primary) bucket
and the compact tables one, or a half.

Input was taken from a set of uniformly drawn unique keys,
taking duplicates from them as necessary.
As in the benchmarks of \cite{awad_analyzing_2023},
to keep the runtime manageable,
we vary the permutations (associating keys to buckets) between measurements,
instead of the input keys themselves.

We benchmark the largest table that can be stored on our GPU.
Apart from the table,
we also need to store input keys,
and reserve memory for storing a.o. return values ($\FOUND/$, $\PUT/$, et cetera).
With the 24GB memory limit of our RTX 4090,
we end up with cuckoo tables of $2^{27}$ slots,
and iceberg tables of $2^{27}$ primary and  $2^{24}$ secondary slots.

\begin{remark}
Larger tables are possible with the use of
CUDA unified memory \cite[section 19.1.2]{nvidia_cuda_guide},
which allows
for allocating more GPU memory than available,
dynamically swapping memory to and from the host (CPU) RAM.
Using unified memory,
our RTX 4090 can work with tables of $2^{29}$ primary and $2^{26}$ secondary slots.
In this case,
the compact tables are $10\times$ faster than the non-compact ones.
However, this is not a particularly fair comparison
(the non-compact tables take at least twice as much memory,
so they will require more swapping)
we focus on the situation where the whole experiment fits in GPU memory.
\end{remark}

If each primary bucket contains 32 of the $2^{27}$ primary slots,
then there are $2^{22}$ primary buckets,
and so compactness will shave 22 bits off of every key.
If each primary slot is 16 bits wide,
they can store remainders of at most 15 bits in length
(one bit is required to indicate whether the slot is occupied or not),
and so  the primary level can store keys of at most
$22 + 15 = 37$ bits.
So we drew our input keys uniformly from $[0,2^{37})$.

We benchmarked cuckoo tables of 8, 16, and 32 slots per bucket,
with slot sizes of 32 (compact) and 64 (non-compact) bits.
For iceberg,
we benchmarked tables with 8, 16, and 32 slots per primary bucket%
---the secondary buckets having half the slots of the primary ones---%
for 16 bit primary slots with 32 bit secondary slots (compact),%
\footnote{
	We use 32 bit secondary slots with the 16 bit primary slots
	because the secondary slots have less compactness
	(as there are fewer secondary slots),
	and so the keys of width 37 would not fit in 16 bit secondary slots.
	They do fit in the 32 bit slots.
}
32 bit primary slots with 32 bit secondary slots (compact),
and 64 bit primary slots with 64 bit secondary slots (non-compact).

We reiterate that each cuckoo resp. iceberg table has the same number of slots,
we only vary how they are divided into buckets and how much of the compactness
is realized in practice
(how much memory each slot consumes).
For cuckoo, the compact 32 bit versions use half the memory of the non-compact 64 bit versions.
For iceberg, the compact 32 bit versions use half the memory of the non-compact 64 bit versions,
and the 16 bit versions use $\frac 9 {32}$ of the memory.

With other recent GPUs,
the RTX 3090 and L40s,
we have found results similar to the ones presented here.
On the older RTX 2080 Ti, we have not.%
\footnote{On the RTX 2080 Ti, compactness shows only a negligible performance increase.}

\subsection{Results}
\subsubsection{Insertion}
\begin{figure}[btp]
	\includegraphics{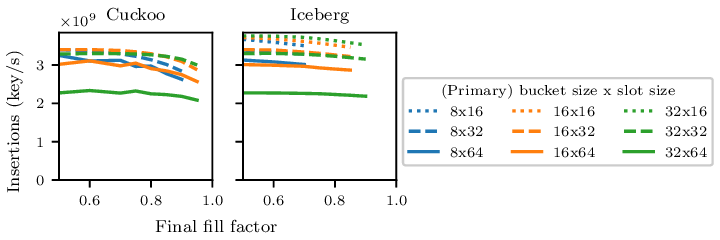}
	\caption{\pr{Put} benchmark.
	The colors indicate bucket setup.
	The solid lines are non-compact tables (baseline),
	the \textbf{dashed lines} use half the memory (compact),
	and the \textbf{dotted lines} use roughly a quarter ($\frac 9 {32}$) the memory (extra compact).
	The compact versions enjoy about 10 times the throughput of the non-compact ones.}
	\label{fig:put}
\end{figure}
Figure~\ref{fig:put}
shows the average throughput of filling the table to a certain fill factor with a batch of unique keys,
measured for fill factors 0.5, 0.6, 0.7, 0.75, 0.8, 0.85, 0.9, and 0.95.
(A static insertion benchmark.)
The cuckoo tables with bucket size 16, 32 achieve fill factor 0.95.
The highest fill factor achieved by the iceberg tables is 0.9,
with bucket size 32.
Both tables show similar performance,
with a modest (5--10\%) advantage for most compact 32 bit versions over 64 bit ones.
The more compact 16 bit slots show an additional slight improvement,
resulting in a 15--20\% advantage over the 64 bit ones.

For the variants with 32 slots per (primary) bucket,
there is a larger speedup.
This is especially relevant for the iceberg table,
as this only reaches fill factor 0.9 on this variant.
The most compact (16 bit primary slots) version of this table is the fastest table in this benchmark,
with a 60\% performance increase over the non-compact version.

\begin{figure}[btp]
	\includegraphics[trim={0 1mm 0 1mm},clip]{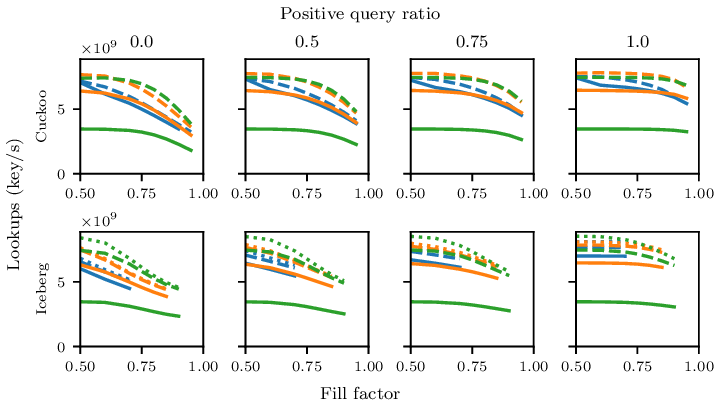}
	\caption{\pr{Find} benchmark. Legend as in Figure~\ref{fig:put}.}
	\label{fig:find}
\end{figure}

\subsubsection{Lookup}
Figure~\ref{fig:find}
shows \pr{Find} benchmark results,
in which for each measurement the table is filled to a certain fill factor,
and then queried for half as many unique keys as there are slots in the table.
The tables with 32 slots per (primary) bucket
benefit significantly from compactness,
roughly doubling throughput over the non-compact version.
For most other bucket sizes,
compactness gives a rough 15--25\% performance increase.

The best cuckoo table (compact, 32 slots per bucket)
outperforms the best compact iceberg table (compact, 32 slots per bucket)
by about 20--30\% at higher load factors,
especially when querying mostly keys that are not in the table.
This can be explained by the fact that with iceberg tables,
both secondary buckets are inspected if the primary bucket is full and does not contain $k$,
while with cuckoo tables it can also happen that only two buckets are inspected (cf. section~\ref{sec:cuckoo-impl}).

\subsubsection{Find-or-put}
We conducted a find-or-put benchmark,
measuring throughput for various combinations of before and after fill factors.
Before each measurement,
the table was filled to the before fill factor.
The $\pr{Fop}$ operation was then issued with as many input keys as there are slots in the table
($2^{27}$ for cuckoo, $2^{27} + 2^{24}$ for iceberg),
containing a mix (with duplicates)
of keys not in the table and
keys already in the table,
such that after the operation,
the table was filled exactly to the after fill factor.

To have a baseline for the iceberg table,
we implemented an unsophisticated
find-or-put operation for the cuckoo table
that first sorts the input keys to detect duplicates
(using the radix sort in Thrust, a library included with the CUDA toolkit),
issues a \pr{Find} once per unique key,
and then inserts one of each key not in the table with \pr{Put}.

\begin{figure}[btp]
	\centering
	\includegraphics[trim={0 1mm 2.3cm 1mm},clip]{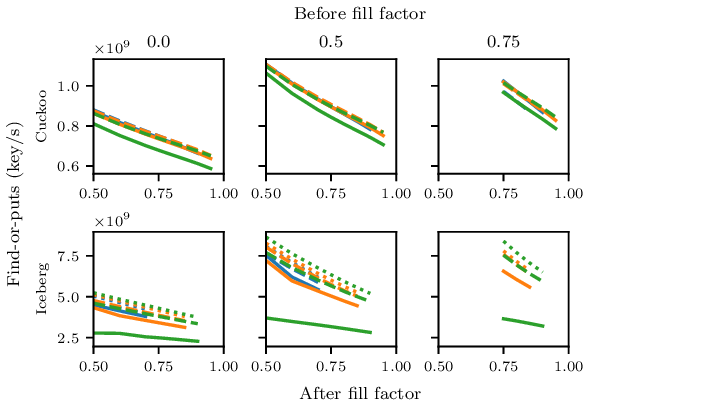}
	\caption{
		\pr{Fop} (find-or-put) benchmark.
		Legend as in Figure~\ref{fig:put}.
		Note the scale difference.
}
	\label{fig:fop}
\end{figure}

Figure~\ref{fig:fop} shows the results.
There is little difference between the cuckoo versions.
The most compact 16 bit iceberg tables show a 10--20\% speedup over the non-compact 64 bit variants,
with a greater 60--100\% speedup for the variant with 32 slots per primary bucket.
The \pr{Find-or-put} of the compact iceberg tables is more than 5 times faster
than the unsophisticated cuckoo find-or-put.

\subsection{Experiments with real-world data}
Apart from synthetic data,
we have also benchmarked the find-or-put operation against real-world data
from a model checking application (HAVi from \cite{wijs_model_checking_2023}).
A single-threaded model checker was used to explore the model,
and the sequence in which the nodes were visited forms our benchmark data.

The set contains about $2^{26.1}$ keys of width 24,
about $2^{23.9}$ without duplicates.
We measured the throughput of handling certain ratios of the data,
for cuckoo tables of $2^{24}$ slots and iceberg tables of $2^{24} + 2^{21}$ slots.

See Figure~\ref{fig:havi} for the results.
The find-or-put of the fastest compact iceberg table is 8 times faster than
the fastest baseline cuckoo find-or-put implementation on the RTX 4090.
The most compact iceberg tables are competitive to their non-compact versions,
with around 5--15\% increase in throughput.

\begin{figure}[btp]
	\centering
	\includegraphics[trim={0 1mm 4.9cm 1mm},clip]{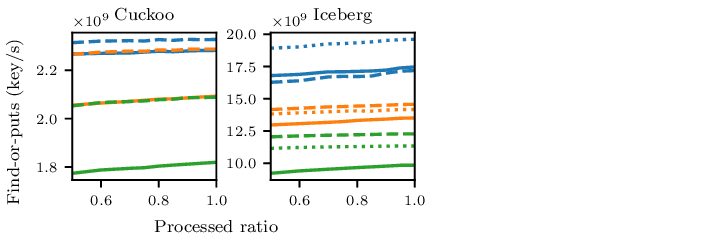}
	\caption{
		\pr{Fop} benchmark for the HAVi dataset.
		Legend as in Figure~\ref{fig:put}.
		The throughput axes do not start at 0,
		so the differences in throughput are smaller than they appear.
	}
	\label{fig:havi}
\end{figure}

\section{Conclusion}
On the GPU,
compact hashing through quotienting
not only saves precious GPU memory,
it also modestly improves performance.
(Compact) iceberg hashing provides a valid alternative to cuckoo hashing,
with comparable $\pr{Find}$ and $\pr{Put}$ performance,
while supporting an efficient, correct find-or-put%
---though our version with only two levels does not achieve as high load factors.
A third level is worth considering,
perhaps made of the slab lists of \cite{pandey_iceberght_2023}.

While we focused on modeling sets,
similar techniques can be used to model key-value dictionaries.
As this increases cache line strain
(when storing the value with the key),
compactness might be of even more importance here.

\bibliographystyle{splncs04}
\bibliography{references}

\newpage
\appendix
\section{Correctness of iceberg find-or-put}\label{app:proof}
We now consider the correctness of the iceberg find-or-put algorithm.
To ease our analysis,
we first introduce a simplified version of Algorithm~\ref{alg:iceberg-fop-abstract}
by making explicit its well-ordering of the slots available to a given key.

Recall that for each key $k$ there are three (non-exclusive) storage locations:
the primary bucket $T_0[a_0(k)]$,
containing $B_0$ slots,
and two secondary buckets $T_1[a_1(k)]$, $T_1[a_2(k)]$,
containing $B_1$ slots each.
Algorithm~\ref{alg:iceberg-fop-abstract} prefers to store $k$ in the first nonempty row in its primary bucket.
If the primary bucket is full, then it opts for the first nonempty slot in the least full secondary bucket instead
(if they are equally full, in the second one).
This corresponds with the following.

\begin{definition}\label{def:order}
	Let $S = \{ 0 \} \times \{ y \mid 0 \leq y < B_0 \} \cup \{ 1, 2 \} \times \{ y \mid 0 \leq y < B_1 \}$.
	Define the well-order $\prec$ on $S$ by
	\begin{itemize}
		\item $(0, x) \prec (1+y,z)$ for all $x,y,z$
		\item $(x, y) \prec (x, z)$ if and only if $y < z$
		\item $(1, x) \prec (2,y)$ if and only if $x < y$
	\end{itemize}
\end{definition}

Given a key $k$,
we see pairs of the form $(0, y)$ as indicating primary slots $T_0[a_0(k)][y]$,
and see the pairs $(1+x, y)$ as slots $T_1[a_x(k)][y]$
in the first ($x = 0$) and second $(x=1)$ secondary buckets.
This essentially orders the slots ``top to bottom, right-to-left'' as seen in Figure~\ref{fig:slot-ordering}.
The insertion behavior of Algorithm~\ref{alg:iceberg-fop-abstract} can be summarized as
continuously attempting to insert $k$ into the first nonempty slot in this ordering until it succeeds.

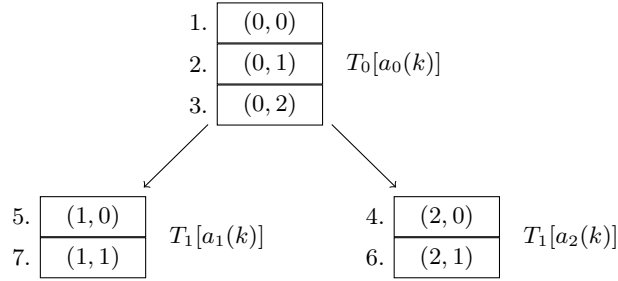
\begin{figure}
	\center
\begin{tikzpicture}
	[ bucket/.style = {matrix of nodes, nodes={draw, minimum width=14mm}}
	, blabel/.style = {}
	]
	\matrix (B0) [bucket] {$(0, 0)$\\$(0, 1)$\\$(0, 2)$\\};
	\matrix (B1) [bucket, below left=1cm of B0] {$(1, 0)$\\$(1, 1)$\\};
	\matrix (B2) [bucket, below right=1cm of B0] {$(2, 0)$\\$(2, 1)$\\};
	\draw[->] (B0) -- (B1);
	\draw[->] (B0) -- (B2);

	\node [right=.8mm of B0] {$T_0[a_0(k)]$};
	\node [right=.8mm of B1] {$T_1[a_1(k)]$};
	\node [right=.8mm of B2] {$T_1[a_2(k)]$};

	\node [left=0mm of B0-1-1] {1.};
	\node [left=0mm of B0-2-1] {2.};
	\node [left=0mm of B0-3-1] {3.};
	\node [left=0mm of B1-1-1] {5.};
	\node [left=0mm of B1-2-1] {7.};
	\node [left=0mm of B2-1-1] {4.};
	\node [left=0mm of B2-2-1] {6.};
\end{tikzpicture}
\caption{%
	Visual representation of order $\prec$ on the primary and secondary slots for a key $k$.
	On top is the primary bucket
	(with $B_0 = 3$ slots),
	below it are the two secondary buckets
	(with $B_1 = 2$ slots each).
	On each slot is drawn the element of $S$ identifying it,
	and to the left of it is its index in the well-ordering.
}
\label{fig:slot-ordering}
\end{figure}

Recall that keys $k$ are stored in their primary bucket by writing $r_0(k)$ into one of its slots;
that in the first secondary bucket,
$k$ is stored by writing $(r_1(k), 0)$ into a slot;
and in the second secondary bucket by writing $(r_2(k), 1)$ into a slot.
We ease notation by using the sign function
$\sg(x)$---which returns $0$ for $0$ and $1$ for all other inputs---%
and functions $r'_i(k)$, where $r'_0(k) = r_0(k)$ and $r'_{1+i}(k) = (r_i(k), i)$.
Given $(x,y) \in S$, the $(x,y)$th slot for $k$ can then be written as
$T_{\sg(x)}[a_x(k)][y]$,
and it contains $k$ if and only if it contains $r'_x(k)$.

In order to state anything about the correctness of the find-or-put algorithm,
we first formalize what table properties it expects and preserves.

\begin{definition}
	A table $T = (T_0, T_1)$ is \emph{well-formed}
	if the following holds:
	\begin{itemize}
		\item Each slot in $T_{\sg(x)}[a_y(k)]$ is either \EMPTY/ or of the form $r'_x(k)$ for some $k$
	\item For each slot of the form $T_{\sg(x)}[a_x(k)][y] = r'_x(k)$, the following holds:
		\begin{align*}
			\forall (x', y') \prec (x, y): T_{\sg(x')}[a_{x'}(k)][y'] \notin \{ \EMPTY/, r'_{x'}(k) \}
		\end{align*}
		That is, if $k$ is in a slot, then all earlier slots for $k$
		are nonempty and do not contain $k$.
		We call this the \emph{order property}.
	\end{itemize}
\end{definition}

The first part states that each table slot must either be empty or contain an appropriate remainder.
The second part (the order property) will play an important role in the correctness argument because of the following observation.

\begin{lemma}
	Well-formed tables do not contain duplicate entries.
\end{lemma}
\begin{proof}
	Let $T$ be a well-formed table,
	and assume towards a contradiction that a key $k$ appears in two distinct slots.
	This means that
	there are $(x,y), (x',y') \in S$
	with $T_{\sg(x)}[a_x(k)][y] = r'_x(k)$
	and $T_{\sg(x')}[a_{x'}(k)][y'] = r'_{x'}(k)$.
	Assume without loss of generality that $(x',y') \prec (x,y)$.
	(We may do this because $\prec$ is a well-order.)
	Applying the order property,
	we find $T_{\sg(x')}[a_{x'}(k)][y'] \neq r'_{x'}(k)$: contradiction.
\end{proof}

Instead of proving the correctness of Algorithm~\ref{alg:iceberg-fop-abstract},
we shall prove it for Algorithm~\ref{alg:fop-short},
which incorporates our new terminology.
It can be obtained from Algorithm~\ref{alg:iceberg-fop-abstract} by
merging the two loops
(trading performance for simplicity)
and realizing that at any point in the algorithm,
the slot it selects to attempt to insert into
is precisely the least nonempty slot (in the local bucket snapshots) according to $\prec$.
From a correctness perspective,
the algorithms are indistinguishable.
We merely choose to analyze Algorithm~\ref{alg:fop-short} for legibility.

\begin{algorithm}
	\caption{Simplified find-or-put. Only $T$ is a shared variable.}
	\label{alg:fop-short}
\begin{pseudo}[kw]*
	\hd{Fop}(k)\\
	while \TRUE/ \\+
		*&\ct{Update local snapshots of buckets}\\
		$b_0 \gets T_0[a_0(k)]$ \label{line:snapshot-0}\\
		$b_1 \gets T_1[a_1(k)]$ \label{line:snapshot-1}\\
		$b_2 \gets T_1[a_2(k)]$ \label{line:snapshot-2}\\
		*\\
		if $(\exists (x,y) \in S : b_x[y] = r'_x(k))$\\+
			return \FOUND/ \label{line:found} \\-
		else if $(\lnot \exists (x,y) \in S : b_x[y] = \EMPTY/)$ \label{line:short-full}\\+
return \FULL/\\-
		else \\+
			$(x,y) \gets \min_\prec \{ (x,y) \in S \mid b_x[y] = \EMPTY/ \}$ \label{line:min}\\
			if $(\cas(T_{\sg(x)}[a_{x}(k)][y], \EMPTY/, r'_{x}(k)))$ \label{line:cas}\\+
				return \PUT/
\end{pseudo}
\end{algorithm}

Algorithm~\ref{alg:fop-short} in words:
make local copies,
\concept{snapshots} $b_0,b_1,b_2$,
of the three buckets for $k$.
Inspect the snapshots:
if any of their slots contain $k$,
return \FOUND/.
Otherwise, try to insert $k$.
If the snapshots do not contain any empty slots,
$k$ cannot be inserted and we return \FULL/.
Otherwise, find the minimal $(x,y) \in S$ under the $\prec$-order
so that slot $b_x[y]$ is empty,
and attempt to insert $k$ into the corresponding slot in the actual table $T$
using an atomic compare-and-swap.
If this succeeds, return \PUT/.
If not, update the snapshots and repeat.

At the end of this section, we shall prove the following theorem,
which also holds for Algorithm~\ref{alg:iceberg-fop-abstract}.

\begin{samepage}
\begin{theorem}\label{thm:correctness}
	Let
	$\Fop(k_0) \parallel \dots \parallel \Fop(k_n)$
	be a parallel composition of Algorithm~\ref{alg:fop-short}
	applied to a well-formed table $T$.
	The following holds:
	\begin{enumerate}\renewcommand{\labelenumi}{\textup{(\alph{enumi})}}
		\item The composition preserves the well-formedness of $T$
		\item Each $\Fop(k_i)$ terminates, and returns either \FOUND/, \FULL/, or \PUT/
		\item $\Fop(k_i)$ returns \FOUND/ or \PUT/ if and only if $k$ is in $T$ after completion
		\item $\Fop(k_i)$ returns \PUT/ if and only if it has inserted $k$ into $T$
		\item If $\Fop(k_i)$ returns \FULL/ then the buckets for $k$ are full and do not contain $k$
	\end{enumerate}
\end{theorem}
\end{samepage}

When we say that an algorithm preserves an invariant $\phi$,
we mean that, for any instruction in the algorithm,
if $\phi$ holds before the instruction,
then it also holds after the instruction.
As the one operation modifying shared state in Algorithm~\ref{alg:fop-short} is atomic,
$\phi$ is then also preserved by parallel compositions
$\Fop(k_0) \parallel \dots \parallel \Fop(k_n)$
of the algorithm.

From now on, we work in a context where there are only
(and finitely many) parallel executions of Algorithm~\ref{alg:fop-short},
that is,
that there are no other algorithms run alongside it.%
\footnote{
	With extra care,
	one can show that the results still hold if other algorithms are executed in parallel with the find-or-put algorithm,
	so long as they preserve well-formedness and do not modify non-empty table slots.
}
This allows us to use the following invariant.

\begin{lemma}\label{lem:monotonicity}
	Algorithm~\ref{alg:fop-short} does not modify non-empty slots in $T$.
	Consequently,
	the following is an invariant for Algorithm~\ref{alg:fop-short}:
	if a slot in a local snapshot is nonempty,
	then it agrees with the table.
	Symbolically:
	$$\forall (x,y) \in S: b_x[y] \neq \EMPTY/ \implies b_x[y] = T_{\sg(x)}[a_x(k)][y].$$
\end{lemma}
\begin{proof}
	The first part is immediate,
	as the only write-operation in Algorithm~\ref{alg:fop-short}
	is a compare-and-swap against \EMPTY/ slots.
	Hence, for the second part,
	we need only consider the lines in Algorithm~\ref{alg:fop-short}
	which modify one of the $b_i$:
	lines \ref{line:snapshot-0}, \ref{line:snapshot-1}, and \ref{line:snapshot-2}.
	These only update the snapshots with their corresponding buckets in $T$,
	so this trivially preserves the property.
\end{proof}

\begin{lemma}\label{lem:well-formed}
	Algorithm~\ref{alg:fop-short} preserves the well-formedness of $T$.
\end{lemma}
\begin{proof}
	We focus on the order property,
	as the other parts of well-formedness are easily seen to be preserved.
	The only line of concern is line~\ref{line:cas},
	as this is the only line potentially modifying $T$.
	If the $\cas$ is unsuccessful,
	the table is not modified and so the well-formedness is trivially preserved.
	If the $\cas$ is successful,
	only the $(x,y)$th slot for $k$ is modified.
	In this case, it thus suffices to prove that
	$\forall (x', y') \prec (x, y): T_{\sg(x')}[a_{x'}(k)][y'] \notin \{ \EMPTY/, r'_{x'}(k) \}$
	is true afterwards.
	By merit of line~\ref{line:min}
	(and the fact that $(x,y)$ are local variables),
	$(x,y)$ satisfies
	$\forall (x', y') \prec (x, y): b_{x'}[y'] \notin \{ \EMPTY/, r'_{x'}(k) \}$.
	Applying Lemma~\ref{lem:monotonicity},
	we find
	$\forall (x', y') \prec (x, y): T_{\sg(x')}[a_{x'}(k)][y'] \notin \{ \EMPTY/, r'_{x'}(k) \}$
	as desired.
	Thus, regardless of the $\cas$ result,
	line~\ref{line:cas} preserves the order property.
\end{proof}

\begin{lemma}\label{lem:termination}
	(In the context of a finite parallel composition,)
	Algorithm~\ref{alg:fop-short} eventually terminates.
	It returns either \FOUND/, \FULL/, or \PUT/.
\end{lemma}
\begin{proof}
	The difficult part is to prove that it must eventually return.
	For this, it suffices that the number
	$|\{ (x, y) \in S : b_x[y] \neq \EMPTY/ \}|$
	of nonempty local snapshot slots grows strictly over the while-loop iterations:
	as there are finitely many slots,
	algorithm then cannot loop infinitely.
	To prove this, note:
	\begin{itemize}
		\item once a snapshot slot is filled, it remains filled by Lemma~\ref{lem:monotonicity};
		\item if during a loop iteration it holds that $k$ is in some snapshot,
			then the algorithm returns \FOUND/;
		\item if during a loop iteration it holds that all snapshot rows are nonempty and do not contain $k$,
			then the algorithm returns \FULL/.
	\end{itemize}
	So at the start of the $m+1$th loop iteration,
	we may conclude that lines \ref{line:min} and \ref{line:cas} were executed in the iteration before.
	It then follows from line~\ref{line:min} and the fact that the snapshots have not been updated since its execution,
	that
	$(x,y) = \min_\prec \{ (x,y) \in S \mid b_x[y] = \EMPTY/ \}$ holds.
	However, after line~\ref{line:cas}---no matter whether the compare-and-swap succeeded---%
	it must have been the case that $T_{\sg(x)}[a_x(k)][y]$ was nonempty.
	By Lemma~\ref{lem:monotonicity},
	this is still the case when the snapshots are updated,
	thus the number of nonempty snapshot slots increases.
\end{proof}

We are now ready to prove Theorem~\ref{thm:correctness}.
\begin{proof}[Theorem~\ref{thm:correctness}]
	Part (a) and (b) follow from Lemma~\ref{lem:well-formed}
	and Lemma~\ref{lem:termination} respectively.
	Part (d) is immediate from inspecting line~\ref{line:cas}.
	One half of part (c) follows from part (d).
	For the other half,
	note that the algorithm only returns \FOUND/ if $k$ is found in a local snapshot
	(line~\ref{line:found}),
	and by Lemma~\ref{lem:monotonicity},
	this means that $k$ must be in $T$ after completion.
	Finally part (e) follows from line~\ref{line:short-full}
	and Lemma~\ref{lem:monotonicity}.
\end{proof}

The above results can be proved for Algorithm~\ref{alg:iceberg-fop-abstract} much the same way as for Algorithm~\ref{alg:fop-short},
with some extra work because of the two separate $\cas$ operations
and to note that
the choice of which slot to $\cas$ into corresponds with the least nonempty slot in the local snapshots
under the $\prec$-order.

\begin{proposition}
	Lemmas \ref{lem:monotonicity},\ref{lem:well-formed}, \ref{lem:termination}, and
	Theorem~\ref{thm:correctness}
	also hold for Algorithm~\ref{alg:iceberg-fop-abstract}.
\end{proposition}

In summary,
any finite parallel composition of find-or-put procedures behaves as desired,
preserving well-formedness,
reporting found keys,
and inserting missing keys where possible.
As a corollary of the well-formedness of $T$,
each key is inserted at most once.
In particular, if multiple processes find-or-put $k$,
at most one returns \PUT/ (and if this happens, the others return \FOUND/).

\end{document}